\documentclass[12pt,a4paper]{article}

\usepackage{amsmath,amsthm,amssymb}
\usepackage{microtype}
\usepackage{nowidow}
\usepackage{xcolor}
\usepackage{natbib}
\usepackage[utf8]{inputenc}
\usepackage{csquotes}
\usepackage{multirow}
\usepackage{float}
\usepackage{mathtools}
\usepackage{dsfont}
\usepackage{graphicx}
\usepackage{longtable}
\usepackage[symbol]{footmisc}
\usepackage[ruled,vlined]{algorithm2e}

\usepackage{amsfonts}

\usepackage{tikz}

\graphicspath{ {./Grafiken/} }

\DeclareMathOperator{\Cov}{Cov}

\usepackage{booktabs,array}

\newcount\rowc

\makeatletter
\def\ttabular{%
\hbox\bgroup
\let\\\cr
\def\rulea{\ifnum\rowc=\@ne \hrule height 1.3pt \fi}
\def\ruleb{
\ifnum\rowc=1\hrule height 1.3pt \else
\ifnum\rowc=6\hrule height \heavyrulewidth 
   \else \hrule height \lightrulewidth\fi\fi}
\valign\bgroup
\global\rowc\@ne
\rulea
\hbox to 10em{\strut \hfill##\hfill}%
\ruleb
&&%
\global\advance\rowc\@ne
\hbox to 10em{\strut\hfill##\hfill}%
\ruleb
\cr}
\def\endttabular{%
\crcr\egroup\egroup}

\title{Cluster-Robust Estimators for Bivariate Mixed-Effects Meta-Regression}
\author{Thilo Welz\footnote{Correspondence: thilo.welz@tu-dortmund.de}, Wolfgang Viechtbauer, Markus Pauly}
\date{\today}

\newcommand{\bs}{\boldsymbol}

\newtheorem{theorem}{Theorem}
\newcommand{\R}{\mathbb{R}}

\begin{document}

\maketitle
\thispagestyle{empty}
\begin{abstract}
Meta-analyses frequently include trials that report multiple effect sizes based on a common set of study participants. These effect sizes will generally be correlated. Cluster-robust variance-covariance estimators are a fruitful approach for synthesizing dependent effects. However, when the number of studies is small, state-of-the-art robust estimators can yield inflated Type 1 errors. We present two new cluster-robust estimators, in order to improve small sample performance. For both new estimators the idea is to transform the estimated variances of the residuals using only the diagonal entries of the hat matrix. Our proposals are asymptotically equivalent to previously suggested cluster-robust estimators such as the bias reduced linearization approach. We apply the methods to real world data and compare and contrast their performance in an extensive simulation study. We focus on bivariate meta-regression, although the approaches can be applied more generally.
\end{abstract}



\newpage
\section{Introduction}

In psychometric and medical research, studies frequently report multiple dependent outcomes. These effects can be synthesized across studies, while incorporating study level moderators, via multivariate meta-regression \citep{berkey1998meta}. This is a more sophisticated approach than averaging the effects within studies to create aggregate effects, which are then synthesized.  A fruitful approach to achieve reliable inference in the case of a multivariate meta-regression is to use a cluster-robust (CR) variance-covariance estimator \citep{hedges2010robust}. Robust estimators are designed to account for potential model misspecification. They can handle dependent effect size estimates and heteroscedastic model errors. A frequent problem in multivariate meta-analysis models is that it is difficult to impossible to compute the variance-covariance matrix of the vector of effect estimates. This is because trials frequently report neither the sampling covariances between study effects nor individual patient data (IPD). This is where CR estimators come into play: They have multiple advantages, such as providing consistent standard errors and asymptotically valid tests without requiring restrictive assumptions regarding the (correlation) structure of the model errors.

Cluster-robust estimators are an extension of heteroscedasticity consistent $(HC)$ estimators. $HC$ estimators, proposed by \cite{white1980heteroskedasticity} and later extended in \cite{cribari2004asymptotic} and \cite{cribari2007inference}, were first proposed in the meta-analytic literature by \cite{sidik2005note}. They have been examined and applied for use in ANCOVA \citep{zimmermann2019}, ordinary least squares regression \citep{hayes2007} and  mixed-effect meta-regression \citep{hedges2010robust,viechtbauer2015comparison,welz2020simulation}. When trials report multiple effects stemming from the same study participants, their clustered, i.e. correlated nature should be accounted for. This is where CR estimators come in. The original formulations of both HC and CR estimators have been shown to possess a downward bias for variance components, as well as yielding highly inflated Type 1 errors of respective test procedures in case of a small number of studies/clusters \citep{viechtbauer2015comparison,tipton2015smallPustejovsky,welz2020simulation}. Therefore it is recommended to instead use one of various improvements that have been suggested. We discuss some of these, such as the bias reduced linearization approach and $CR_3$ as introduced in \cite{bell2002bias}, as well as two new proposals in the chapter on cluster-robust estimators. These can be applied generally for multivariate meta-regression, but we focus specifically on the bivariate case.


First, we present the statistical model, as well as tests and confidence regions for the model coefficients in Section 2. In Section 3, we describe multiple CR estimators, including two new suggestions. In Section 4, we conduct a real world data analysis. Section 5 describes the design and results of our simulation study. We close with a discussion of the results and an outlook for future research (Section 6).

\section{The Set-up}
\label{sec:setup}

The usual multivariate mixed-effects meta-regression model \citep{jackson2011multivariate} is given by

\begin{equation}
\label{model_classic}
\boldsymbol{Y_i} = \boldsymbol{X_{i} \beta + u_i + \varepsilon_i}, \ i=1,\ldots,l,
\end{equation}
where $k$ is the number of independent studies, $\bs{\beta} \in \mathbb{R}^q$ is a vector of coefficients and $\bs{X_i}$ a $p_i \times q$ design matrix of study-level covariates. In the following we will assume that there are $p$ effects of interest per study, but only $p_i \leq p$ effects are observed (reported) in study $i$, i.e. $\bs{Y_i} \in \mathbb{R}^{p_i}$. Furthermore, $\boldsymbol{u_i}$ is a random effect that is typically assumed to be multivariate normally distributed with $\boldsymbol{u_i} \sim \mathcal{N}(\boldsymbol{0,T_i})$ and $\boldsymbol{\varepsilon_i}$ is the \textit{within-study} error with $\boldsymbol{\varepsilon_i} \sim \mathcal{N}(\boldsymbol{0,V_i})$. With $\boldsymbol{T_i}$ we refer to the $p_i \times p_i$ submatrix of the matrix $\boldsymbol{T} = \begin{pmatrix}
\tau_1^2 & \tau_{12}\\
\tau_{12} & \tau_2^2
\end{pmatrix}$, denoting the $p \times p$ \textit{between}-study variance-covariance matrix (under complete data). $\bs{V_i}$ refers to the corresponding $p_i \times p_i$ \textit{within}-study variance-covariance matrix.
A typical example would be a compound symmetry structure for $\bs{T_i}$, see Section \ref{sim_study} below. We rewrite model (\ref{model_classic}) in matrix notation as

\begin{equation}
\boldsymbol{Y = X\beta + u + \varepsilon},
\end{equation}

\noindent
with $\bs{\beta} \in \mathbb{R}^q$, $\bs{Y = (Y_1',\ldots,Y_K')'}$, and design matrix $\textbf{X}$. Assuming that we have a block diagonal matrix of weights $\boldsymbol{\widehat{W}} = \text{diag}(\boldsymbol{\widehat{W}}_1,\ldots,\boldsymbol{\widehat{W}}_K)$, usually corresponding to the inverse variance weights with $\boldsymbol{\widehat{W}_i} = \left(\boldsymbol{\widehat{T}_i+V_i}\right)^{-1}$, then the weighted least squares estimator for $\boldsymbol{\beta}$ is given by \citep{hedges2010robust}

\begin{equation}
\boldsymbol{\hat{\beta} = (X'\widehat{W}X)^{-1}X'\widehat{W}Y}.
\end{equation}


\noindent
We will focus on constructing (multivariate) confidence regions for $\boldsymbol{\beta}$ and confidence intervals for the individual coefficients $\beta_j, \ j=1,\ldots,q$ based on testing the hypotheses $H_0: \{\boldsymbol{\beta} = \boldsymbol{\beta}_0\}$ vs. $H_1: \{\boldsymbol{\beta} \neq \boldsymbol{\beta}_0\}$. We set $\bs{\Sigma} = \text{Cov}(\bs{\hat{\beta}})$ and denote estimates thereof by $\bs{\widehat{\Sigma}}$. We discuss specific choices for estimating $\bs{\Sigma}$ in Section \ref{CR-estimators}.

Neglecting multiplicity, we note that a commonly used confidence interval for $\beta_j, \ j=1,\ldots,q$ is given by

\begin{equation}
\hat{\beta}_j \pm \sqrt{\boldsymbol{\widehat{\Sigma}}_{jj}} z_{1-\alpha/2}.
\end{equation}

\noindent
Here $z_{1-\alpha/2}$ denotes the $1-\alpha/2$ quantile of the standard normal distribution and $\widehat{\bs{\Sigma}}_{jj}$ denotes the $j^{th}$ diagonal element of $\widehat{\bs{\Sigma}}$. A confidence interval with better small sample performance that is asymptotically equivalent for $k \rightarrow \infty$ is given by using the $t_{p(k)-q,1-\alpha/2}$ quantile instead, which refers to the $1-\alpha/2$ quantile of the $t$-distribution with $p(k)-q$ degrees of freedom. Here $p(k) \coloneqq \sum_{i=1}^k p_i$ is the total number of observed effects, which is equal to the number of studies $k$ in the univariate setting \citep{viechtbauer2015comparison}. Alternatively the degrees of freedom of the $t$ distribution can be estimated via a Satterthwaite approximation, as suggested by \cite{bell2002bias}.

In order to construct a $(1-\alpha)$ confidence region for $\boldsymbol{\beta}$ we consider the usual Wald-type test-statistic \citep{tipton2015smallPustejovsky}

\begin{equation}
\label{mv_statistic}
Q = (\boldsymbol{\hat{\beta}} - \boldsymbol{\beta}_0)' \widehat{\boldsymbol{\Sigma}}^{-1} (\boldsymbol{\hat{\beta}} - \boldsymbol{\beta}_0),
\end{equation}

\noindent
Alternatively, if one were interested in testing more general hypotheses of the form $H_0: \{\bs{H} \boldsymbol{\beta} = \boldsymbol{c}\}$ vs. $H_1: \{\bs{H} \boldsymbol{\beta} \neq \boldsymbol{c}\}$ for some hypothesis matrix $\bs{H} \in \mathbb{R}^{s \times q}$ (which we assume to be of full rank) and vector $\textbf{c} \in \mathbb{R}^s$, then the test statistic becomes

\begin{equation*}
\label{mv_statistic_H}
Q_{\bs{H}} = (\bs{H}\boldsymbol{\hat{\beta}} - \boldsymbol{c})' (\bs{H}\widehat{\boldsymbol{\Sigma}}\bs{H}')^{-1} (\bs{H}\boldsymbol{\hat{\beta}} - \boldsymbol{c}),
\end{equation*}

\noindent
For example, the special case of a test regarding a single regression coefficient $\beta_a$ would be given by $\bs{H}$ equal to a vector of length $q$ with a 1 at entry $a$ and 0 otherwise.

Under the null hypothesis $Q$ is approximately $\chi^2_{q}$-distributed (and $Q_{\bs{H}}$ approximately $\chi^2_f$-distributed with $f = \text{rank}(\bs{H})$), assuming $\bs{\Sigma}$ is positive definite. However, it is known that tests based on this approximation can perform poorly for small to moderate values of $k$ \citep{tipton2015smallPustejovsky}. An arguably better alternative is the $F$-test

\begin{equation}
\label{F-test}
\mathds{1}\left\{Q > q F_{q,k-q,1-\alpha}\right\},
\end{equation}

\noindent
where $F_{q,k-q,1-\alpha}$ denotes the $1-\alpha$ quantile of an $F$-distribution with $q$ and $k-q$ degrees of freedom. This is analogous to the $t$-tests for univariate coefficients and is superior to the test based on the asymptotic $\chi^2$-approximation \citep{tipton2015smallPustejovsky}. However, the $F$-test has been criticized for only performing well in certain scenarios \citep{tipton2015small}. As a remedy for smaller $k$, \cite{tipton2015smallPustejovsky} proposed to approximate $Q$ by a Hotelling's $T^2$ distribution with parameters $q$ and (degrees of freedom) $\eta$, such that

\begin{equation}
\label{HotellingsT2}
\frac{\eta - q +1}{\eta q} Q \sim F(q,\eta -q +1).
\end{equation}

\noindent
They discuss different approaches for estimating the degrees of freedom $\eta$. Based on their research, they recommend an estimation approach, which they call \enquote{HTZ}. We briefly summarize this estimator, originally proposed by \cite{zhang2012approximate} for heteroscedastic one-way MANOVA, and refer to their paper for details.

First note that the statistic in (\ref{mv_statistic}) can also be written as $Q = \bs{z}'\bs{S}^{-1}\bs{z}$ with $\bs{z} = \bs{\Sigma}^{-1/2}(\boldsymbol{\hat{\beta}} - \boldsymbol{\beta}_0)$ and $\bs{S} = \bs{\Sigma}^{-1/2} \bs{\widehat{\Sigma}} \bs{\Sigma}^{-1/2}$. Under $H_0$, $\bs{z}$ is normally distributed with mean $\bs{0}$ and covariance $\bs{I}$ \citep{tipton2015smallPustejovsky}. Moreover, if $\bs{S}$ is a random $q \times q$ matrix such that $\eta \bs{S}$ follows a Wishart distribution with $\eta$ degrees of freedom and scale matrix $\bs{I}_q$, the estimator is given by

\begin{equation*}
\hat{\eta}_Z = \frac{q(q+1)}{\sum_{a=1}^q \sum_{b=1}^q \text{Var}(s_{ab})}.
\end{equation*}

\noindent
Here $s_{ab}$ denotes the entry $(a,b)$ of $\bs{S}$. This approach corresponds to setting the total variation in $\bs{S}$ equal to the total variation in a Wishart distribution \citep{tipton2015smallPustejovsky}.

However, our own simulations showed that there are situations when $\hat{\eta}_Z < q - 1$ and therefore $\hat{\eta}_Z-q+1<0$. Specifically this frequently happened in cases with a small number of studies ($k \leq 5$). As the degrees of freedom in an $F$ distribution cannot be negative the HTZ approach is not applicable here. Therefore we will stick to the classical $F$-test (\ref{F-test}), although we propose a small sample adjustment. In our simulations the $F$-test (\ref{F-test}) leads to very liberal or conservative results, depending on the variance-covariance estimator used, in settings with $k=5$ studies. We therefore propose to truncate the denominator degrees of freedom at the value two, i.e. we consider the $F$-test

\begin{equation}\label{F-test-adjust}
    \mathds{1}\left\{Q > q F_{q,\max(2,k-q),1-\alpha}\right\}.
\end{equation}

The simple motivation behind this adjustment is that for an $F_{m,n}$ distribution with degrees of freedom $m$ and $n$ the expected value $\frac{n}{n-2}$ only exists when $n>2$. We also tested a truncation of the denominator degrees of freedom at three. However, simulations indicate superior coverage of respective confidence intervals for a truncation at two.

Confidence regions for $\bs{\beta}$ can be derived via test inversion. For example, if (\ref{F-test-adjust}) is a test for $H_0: \{\bs{\beta = \beta_0}\}$ vs. $H_1: \{\bs{\beta \neq \beta_0}\}$, then the set

\begin{equation}
\label{confidence_region}
\Lambda := \left\{\bs{\beta} \in \mathbb{R}^q: (\boldsymbol{\hat{\beta}} - \bs{\beta})' \widehat{\boldsymbol{\Sigma}}^{-1} (\bs{\hat{\beta}} - \boldsymbol{\beta}) \leq qF_{q,\max(2,k-q),1-\alpha}\right\}
\end{equation}

\noindent
is a corresponding confidence region for $\bs{\beta}$.

A confidence ellipsoid can be obtained following \cite{johnson2014applied}, based on the eigenvalues $\hat{\lambda}_j$ and eigenvectors $\bs{\hat{e}}_j$ of $\widehat{\bs{\Sigma}}$. This means $\Lambda$ is an ellipsoid centered around $\bs{\hat{\beta}}$, whose axes are given by

\begin{equation*}
\bs{\hat{\beta}} \pm \sqrt{\hat{\lambda}_j qF_{q,\max(2,k-q),1-\alpha}}\bs{\hat{e}}_j, \ \ j=1,\ldots,q.
\end{equation*}

This means $\Lambda$ extends for $\sqrt{\hat{\lambda}_j qF_{q,\max(2,k-q),1-\alpha}}$ units along the estimated eigenvector $\bs{\hat{e}}_j$ for $j=1,\ldots,q$. 
Since the volume of an $n$-dimensional ellipsoid with axis lengths $a_1,\ldots,a_n$ is given by \citep{wilson2010volume}

\begin{equation*}
V = \frac{2 \pi^{n/2}}{n \Gamma(n/2)} \prod_{i=1}^n a_i,
\end{equation*}

\noindent
the volume of the confidence ellipsoid $\Lambda$ is equal to

\begin{equation*}
V_{\Lambda} = \frac{2 \pi^{q/2}}{q \Gamma(q/2)} \prod_{i=1}^q \sqrt{\hat{\lambda}_i qF_{q,\max(2,k-q),1-\alpha}}.
\end{equation*}

\section{Cluster-Robust Covariance Estimators}
\label{CR-estimators}


Robust variance-covariance estimators, also known as sandwich estimators or Huber-White estimators, have been recommended as a promising alternative in the context of meta-regression \citep{hedges2010robust,tipton2015small,welz2020simulation}. Robust estimators are designed to account for potential model misspecification. They have many desirable properties, such as  consistency under heteroscedasticity or asymptotic normality \citep{hedges2010robust} without making restrictive assumptions about the specific form of the  effect sizes' sampling distributions.

The reliability of confidence regions based on the statistic (\ref{mv_statistic}) depends on the quality of the estimator $\widehat{\boldsymbol{\Sigma}}$ for $\boldsymbol{\Sigma} = \text{Cov}(\boldsymbol{\hat{\beta}})$. The standard (Wald-type) estimator, which we will refer to as $ST$, is given by $(\boldsymbol{X' \widehat{W}X})^{-1}$. The motivation behind this estimator is that the true covariance matrix of $\bs{\widehat{\beta}}$ (given correct weights) is equal to $\bs{\Sigma} = (\boldsymbol{X' WX})^{-1}$ with $\bs{W} = \text{diag}\left(\bs{W_1},\ldots,\bs{W_K}\right)$ and $\bs{W_i} = \bs{T_i+V_i}$. However, this ignores the imprecision in the estimation of $\bs{T,V}$ and therefore in the estimation of $\textbf{W}$. In fact, if $\bs{T}$ is estimated poorly, this may lead to deviations from nominal Type 1 error and coverage of corresponding confidence regions \citep{sidik2005note}.

In the case of univariate meta-analysis and meta-regression \sloppy heteroscedasticity-consistent (HC) estimators can be applied \citep{sidik2005note,viechtbauer2015comparison,welz2020simulation}. For multivariate meta-regression however, the correlated nature of the study effects needs to be taken into account. We therefore consider cluster-robust (CR) estimators. A selection of CR estimators is, e.g., implemented in the \texttt{R} package \texttt{clubSandwich} \citep{pustejovskyClubSandwich}. The package recommendation is the \enquote{bias reduced linearization} approach $CR_2$, which is discussed in detail in \cite{tipton2015smallPustejovsky,pustejovsky2018small}. Sandwich estimators (of HC- as well as CR-type) are all of the general form


\begin{equation}
\widehat{\boldsymbol{\Sigma}} = \boldsymbol{(X'\widehat{W}X)^{-1}X'\widehat{W} \widehat{\Omega} \widehat{W}X(X'\widehat{W}X)^{-1}},
\end{equation}

\noindent
with the differences lying in the central \enquote{meat} matrix $\widehat{\boldsymbol{\Omega}}$, surrounded by the \enquote{bread}. This form motivates the name \enquote{sandwich} estimator. 
$HC_1$\footnote[2]{$\bs{\widehat{\Sigma}_{HC_1}} = \tfrac{k}{k-q} \boldsymbol{(X'\widehat{W}X)^{-1} \left(\sum_{i=1}^K X_i'\widehat{W}_i \hat{\varepsilon}_i^2 \widehat{W}_iX_i \right) (X'\widehat{W}X)^{-1}}$} is arguably the best known sandwich estimator in the context of univariate meta-regression \citep{hedges2010robust,viechtbauer2015comparison,tipton2015smallPustejovsky}. However, the extensions $HC_3$ and $HC_4$ are frequently recommended as superior alternatives in the non meta-analytic literature, see \cite{cribari2007inference} for details, and have been shown to be superior to $HC_1$ \citep{long2000using,hayes2007,zimmermann2019}. A natural extension of $HC_1$ for the multivariate setting and what we will refer to as $CR_1^*$ is defined as

\begin{equation}
\widehat{\boldsymbol{\Sigma}}_{CR_1^*} = \tfrac{k}{k-q} \boldsymbol{(X'\widehat{W}X)^{-1} \left(\sum_{i=1}^K X_i'\widehat{W}_i \widehat{\Omega}_i \widehat{W}_iX_i \right) (X'\widehat{W}X)^{-1}},
\end{equation}

\noindent
where $\widehat{\bs{\Omega}}_i = \bs{E_iE_i'}$ with $\bs{E_i} = \boldsymbol{Y_i-X_i\hat{\beta}}$ and $\frac{k}{k-q}$ is a correction factor that converges to 1 as $k$ goes to infinity. The motivation for this factor is to correct for a liberal behavior in case of few studies/clusters $k$; see the \texttt{clubSandwich} package for similar choices.

However, as our simulation study below will show, tests based on $CR_1^*$ are still quite liberal when $k$ is small. An alternative is to instead use a bias reduced linearization approach, which was originally proposed by \cite{bell2002bias} and further developed by \cite{pustejovsky2018small}. This estimator, called $CR_2$, is designed to be exactly unbiased under the correct specification of a working model. This is achieved via a clever choice of adjustment matrices in the formulation of the estimator, see \cite{tipton2015smallPustejovsky,pustejovsky2018small} for details. This is the recommended approach in the \texttt{clubSandwich} package \citep{pustejovskyClubSandwich}. Another alternative is the $CR_3$ estimator, which is a close approximation of the leave-one-(cluster)-out Jackknife variance-covariance estimator. $CR_3$ is also implemented in the \texttt{clubSandwich} package.

However, all of the estimators above can be unsatisfactory for small $k$, as our simulations will show. Therefore, in addition to these $CR$-estimators, we propose two others, which are extensions of the $HC_3$ and $HC_4$ estimators. Since $HC_3$ and $HC_4$ often outperform both $HC_1$ and $HC_2$ in the univariate regression setting \citep{long2000using,cribari2004asymptotic,welz2020simulation}, one would suspect their respective cluster-robust extensions to outperform in the case of multivariate regression. We therefore define $CR_3^*$ and $CR_4^*$ via

\begin{equation}
\widehat{\boldsymbol{\Sigma}}_{CR_3^*} = \boldsymbol{(X'\widehat{W}X)^{-1} \left(\sum_{i=1}^K X_i'\widehat{W}_i \widehat{\Omega}_{3i} \widehat{W}_iX_i \right)(X'\widehat{W}X)^{-1}},
\end{equation}

\begin{equation}
\widehat{\boldsymbol{\Sigma}}_{CR_4^*} = \boldsymbol{(X'\widehat{W}X)^{-1} \left(\sum_{i=1}^K X_i'\widehat{W}_i \widehat{\Omega}_{4i} \widehat{W}_iX_i \right) (X'\widehat{W}X)^{-1}}.
\end{equation}

\noindent
Here $\widehat{\bs{\Omega}}_{3i}$ is defined as
\begin{equation}
\label{cr3star}
    \widehat{\bs{\Omega}}_{3i} = \widehat{\bs{\Omega}}_i - \Delta + \Delta \cdot \left(\bs{I}_{p_i}-\text{diag}(\bs{H}_i)\right)^{-2},
\end{equation}

\noindent
where $\bs{H}_i$ refers to the submatrix of $\bs{H}$ with entries pertaining to study $i$, $p_i$ is the number of observed effects in study $i$ and $\Delta = \text{diag}(\bs{E_iE_i'})$. $\bs{H}$ refers to the hat matrix $\bs{H} = \bs{X(X'\widehat{W}X)^{-1}X'\widehat{W}}$. Furthermore, $\widehat{\bs{\Omega}}_{4i}$ is equal to (\ref{cr3star}) except $\Delta$ is multiplied with $\left(\bs{I}_{p_i}-\text{diag}(\bs{H}_i)\right)^{-\delta_i}$, where $\delta_i = \min\left\{ 4,h_{ii}/\bar{h} \right\}$ with $h_{ii}$ denoting the $i$-$th$ diagonal element of $\bs{H}$ and $\bar{h}$ is the average of the values in the diagonal of the hat matrix. This data-dependent exponent stems from the $HC_4$ suggestion by \cite{cribari2004asymptotic}. $HC_4$ performs well in univariate meta-regression \citep{welz2020simulation} and therefore motivates an extension to the cluster-robust context.

We highlight that our proposed estimator $CR_3^*$ is different from the estimator $CR_3$ implemented in the \texttt{R} package \texttt{clubSandwich} as proposed by \cite{bell2002bias}. Whereas the latter uses the entire hat matrix for each cluster, we propose to use just the diagonal elements. In contrast, the \enquote{meat} matrix for $CR_3$ is given by $\sum_{i=1}^K \bs{X_i'\widehat{W}_i (I - H_i)^{-1} \widehat{\Omega}_{i} (I - H_i)^{-1} \widehat{W}_iX_i}$. Furthermore note that $CR_3^*$ is not even equal to the estimator with meat matrix given by $$\sum_{i=1}^K \bs{X_i'\widehat{W}_i (I - \text{diag}(H_i))^{-1} \widehat{\Omega}_{i} (I - \text{diag}(H_i))^{-1} \widehat{W}_iX_i}$$ because $\bs{\widehat{\Omega}_{i}}$ is in general not a diagonal matrix (only block-diagonal), due to the clustered nature of the data.


For univariate regression we were able to prove the asymptotic equivalence of all $HC$ estimators, which is formulated in the supplement of \cite{welz2020simulation}. Under some some weak regularity conditions it follows that the leverages asymptotically converge to zero, as the number of studies $k$ goes to infinity. Therefore, we expected similar results to hold for $CR$ estimators with analogous arguments. A theorem regarding the asymptotic equivalence of $CR$ estimators under regularity conditions is given in the supplement of this paper, along with a proof.

\section{Data Analysis}

We exemplify the methods presented in this manuscript with the analysis of a dataset containing 81 trials examining overall (OS) and/or disease-free survival (DFS) in neuroblastoma patients with amplified (extra copies) versus normal MYC-N genes. The data are contained in the \texttt{R} package \texttt{metafor} and were previously analyzed by \cite{riley2003systematic, riley2007evaluation}. Amplified MYC-N levels are associated with poorer outcomes. The effect measures are log hazard ratios with positive values indicating an increased risk of death or relapse/death for patients with higher MYC-N levels as compared to patients with lower levels. 17 studies reported both outcomes, 25 studies only reported DFS and 39 studies only reported OS.

The dataset contains the log hazard ratios and the corresponding sampling variances. However, since no information is available on the sampling covariances between OS and DFS we must make some assumptions with regard to our working model. In the spirit of a sensitivity analysis we will first assume a weaker correlation of $\varrho_1 = 0.5$ and subsequently a stronger correlation of $\varrho_2 = 0.8$ and then compare the results. This means for a hypothetical study $i$ that reports log hazard ratios for OS and DFS, $y_{i,OS}$ and $y_{i,DFS}$, with an assumed correlation of 0.5 along with respective sampling variances $\sigma_{i,OS}^2$ and $\sigma_{i,DFS}^2$, we have the sampling variance-covariance matrix $V_i = \begin{pmatrix}
\sigma_{i,OS}^2 & 0.5 \cdot \sigma_{i,OS}\sigma_{i,DFS}\\
0.5 \cdot \sigma_{i,OS}\sigma_{i,DFS} & \sigma_{i,DFS}^2
\end{pmatrix}$.

We assume a multivariate meta-regression model that includes a random effect as in Section \ref{sec:setup} as well as an unstructured (but positive definite) variance-covariance matrix. In the following we are interested in testing whether both pooled effects are different from zero. When the full dataset is analyzed, the Wald-test for $H_0: \{\bs{\beta = 0}\}$ vs. $H_1: \{\bs{\beta \neq 0}\}$ returns a p-value $< 0.001$ for all CR estimators and for both $\varrho_1$ and $\varrho_2$. However, let us assume we only had the data from studies 1-5, which all contain results for both OS and DFS. Such a situation is not unrealistic, considering the median number of studies per meta-analyses in a sample of 22,453 published meta-analyses from the Cochrane Database was three \citep{davey2011characteristics}. This reduced dataset is shown in Table \ref{datrileysample}. The p-values for the estimators $CR_1*$, $CR_3*$, $CR_4*$, $CR_2$ and $ST$ for assumed correlations $\varrho_1, \ \varrho_2$ are displayed in Table \ref{p-values-riley}.

\begin{table}[H]
\centering
\begin{tabular}{crrrl}
  \hline
study & $y_i$ & $v_i$ & outcome \\ 
  \hline
  1 & -0.11 & 0.45 & DFS \\ 
  1 & -0.14 & 0.66 & OS \\ 
  2 & 0.30 & 0.07 & DFS \\ 
  2 & 0.67 & 0.08 & OS \\ 
  3 & 0.41 & 0.77 & DFS \\ 
  3 & 0.43 & 0.66 & OS \\ 
  4 & 0.47 & 0.29 & DFS \\ 
  4 & 2.08 & 0.45 & OS \\ 
  5 & 0.76 & 0.24 & DFS \\ 
  5 & 0.70 & 0.31 & OS \\ 
   \hline
\end{tabular}
\caption{Sample of five studies containing log hazard ratios ($y_i$) for disease-free and overall survival and their respective sampling variances ($v_i$).}
\label{datrileysample}
\end{table}

The results show that when the number of studies is small the p-values can vary substantially, depending on the choice of estimator. Furthermore, the results based on CR estimators appear to be more stable and depend much less on the underlying $\bs{V}$ matrix i.e. the assumed correlation between OS and DFS than the standard estimator $\bs{(X'\widehat{W}X)^{-1}}$. This motivates the use of a CR approach over the standard variance-covariance estimator.

\begin{table}[H]
\centering
\begin{tabular}{crr}
\multirow{2}{*}{Estimators} & \multicolumn{2}{c}{p-values} \\
 & $\varrho_1$ & $\varrho_2$ \\ 
  \hline
CR1* & 0.073 & 0.075\\ 
CR3* & 0.069 & 0.077\\ 
CR4* & 0.076 & 0.090\\ 
CR2 & 0.054 & 0.055\\ 
ST & 0.138 & 0.206\\ 
   \hline
\end{tabular}
\caption{p-values of Wald-tests based on CR estimators and the standard variance-covariance estimator $\bs{(X'\widehat{W}X)^{-1}}$ for assumed correlations of $\varrho_1 = 0.5$ and $\varrho_2 = 0.8$.}
\label{p-values-riley}
\end{table}

\section{Simulation Study}
\label{sim_study}

\textbf{Simulation Design} In order to assess the performance of the previously discussed methods, we conducted a Monte Carlo simulation. We considered $k \in \{5,10,20,40\}$ studies, average study sizes $N \in \{40,100\}$ with balanced treatment and control groups, coefficient vectors $\bs{\beta} = (\beta_0,\beta_1,\beta_2,\beta_3)' \in \{(0,0,0,0)',(0.2,0.2,0.1,0.1)',(0.4,0.4,0.2,0.3)'\}$, correlations $\varrho \in \{0,0.3,0.7\}$ and missing data ratios from $\{0,0.1,0.2,0.3,0.4\}$. The latter refers to the number of studies that only report one of the two effects of interest and $\varrho$ refers to the IPD correlations between the two observed outcomes. In the coefficient vector $\boldsymbol{\beta}$ the first two entries refer to the population means of the two effects of interest and the other two represent the effect of the study-level moderator on each effect respectively. Study sizes were varied, such that for an average study size $N$, $20\%$ of studies had size $0.8N,0.9N,\ldots,1.2N$ respectively. Datasets with missing data were generated by first simulating complete data and then removing entries completely at random.

The simulated study-level effects are (correlated) standardized mean differences (SMD). We estimated these SMDs via the adjusted Hedges' $g$ \citep{hedges1981distribution}
\begin{equation*}
g \coloneqq \frac{\Gamma(m/2)}{\sqrt{(m/2)}\Gamma((m-1)/2)} d
\end{equation*}
with $m=n_T+n_C-2$ and where $n_T$ and $n_C$ refer to the treatment and control group sizes. Hedges' $g$ is defined as $d = (\bar{x}_T - \bar{x}_C)/s^*$, with a pooled standard deviation $s^* = \sqrt{\frac{(n_T-1)s_T^2 + (n_C-1)s_C^2}{m}}$, where $s_T^2,s_C^2$ refer to the variances in the treatment and control groups respectively \citep{hedges1981distribution}. This adjustment to Hedges' $g$ yields an unbiased effect estimator \citep{lin2021evaluation}. We generated the SMDs by first simulating individual participant data (IPD). The treatment and control group IPD observations $Y_{ij}^T$ and $Y_{ij}^C$ were drawn from bivariate normal distributions respectively. More precisely, for study $i=1,\ldots,k$ and participant $j=1,\ldots,N_i/2$ the observations are drawn from $Y_{ij}^T \sim \mathcal{N}(\theta_i,P)$ and $Y_{ij}^C \sim \mathcal{N}(\bs{0},P)$ with $\theta_i = \bs{X\beta + u_i}$ and $P = \begin{pmatrix}
1 & \varrho\\
\varrho & 1
\end{pmatrix}$ is the population correlation matrix of the outcomes in study $i$. $\bs{X}$ is a $2 \times q$ design matrix of covariates. In our specific simulation design of a single study-level covariate $x$ with potentially different influence on the two study effects we have $\bs{X} = \begin{pmatrix}
1 & 0 & x & 0\\
0 & 1 & 0 & x
\end{pmatrix}$.

\noindent
For the heterogeneity matrix $\bs{T}$ we consider the two settings $$\begin{pmatrix}
\tau^2 & 0.2\tau^2\\
0.2\tau^2 & \tau^2
\end{pmatrix} \text{and }
\begin{pmatrix}
\tau^2 & 0.4\tau^2\\
0.4\tau^2 & 2\tau^2
\end{pmatrix}.$$

\noindent
For $M = N/2$ (average size of the treatment and control groups), we set $\tau^2 := \frac{2}{M} + \frac{\beta_0^2}{4M} = \frac{4}{N} + \frac{\beta_0^2}{2N}$, which is approximately equal to the sampling variance of the standardized mean difference \citep{borenstein2021introduction}. This corresponds to an $I^2$ value of 0.5. Here, $I^2$ refers to the percentage of the total variation across studies that is due to heterogeneity rather than sampling variation \citep{higgins2002quantifying}.

We briefly discuss the covariance between two SMDs in the setting where we have a single treatment and control group but with different outcome measures. The resulting effect sizes will be correlated because the outcomes are collected from the same study participants. \cite{olkin2009stochastically} showed that a large sample estimate for the covariance between two SMDs $d_1$ and $d_2$ with estimated (raw data) correlation $\hat{\varrho}$ is given by
\begin{equation}
    \widehat{\Cov}(d_1,d_2) = \hat{\varrho}\left(\tfrac{1}{n_T} + \tfrac{1}{n_C}\right) + \frac{\hat{\varrho}^2d_1d_2}{m}.
\end{equation}
Thus we obtain
\begin{equation}
    \widehat{\Cov}(g_1,g_2) = \left(\frac{\Gamma(m/2)}{\sqrt{(m/2)}\Gamma((m-1)/2)} \right)^2 \left(\hat{\varrho}\left(\frac{1}{n_T} + \frac{1}{n_C}\right) + \frac{\hat{\varrho}^2d_1d_2}{m}\right).
\end{equation}    
All results are based on a nominal significance level $\alpha=0.05$. For each scenario we performed $N=5000$ simulation runs. The primary focus was on comparing empirical coverage of the confidence regions (\ref{confidence_region}) with nominal coverage being $1-\alpha=0.95$. For 5000 iterations, the Monte Carlo standard error of the simulated coverage will be approximately $\sqrt{\frac{0.95 \times 0.05}{5000}} \approx 0.31\%$ and assuming a power of 80\% the Monte Carlo standard error of the simulated power will be approximately $\sqrt{\frac{0.8 \times 0.2}{5000}} \approx 0.57\%$ \citep{morris2019using}.

All simulations were performed using the open-source software \texttt{R}. The \texttt{R} scripts written by the first author especially make use of the \texttt{metafor} package for meta-analysis \citep{viechtbauer2010conducting} as well as James Pustejovsky's \texttt{clubSandwich} package.

\begin{center}
\textbf{Results}
\end{center}

Figures \ref{fig:coverage5}--\ref{fig:coverage40} display the empirical coverage based on the adjusted $F$-test (\ref{F-test-adjust}) and estimators $CR_1^*$, $CR_3^*$, $CR_4^*$, $CR_2$ and $ST$. $CR_1^*$ and $CR_2$ yield much less than nominal coverage 95\% in all settings, but especially for $k<40$. $CR_2$ gives around 50\% coverage for five studies, between 70-80\% for ten, 82-87\% for twenty and 88-91\% coverage for forty studies. The $CR_1^*$ estimator yields between 25-50\% coverage for five studies, 65-75\% for ten, 80-86\% for twenty and 87-91\% for forty studies. It is interesting to observe a clustering of coverage results for the estimator $CR_1*$ and $k=5$ (depending on the inter-study correlation of effects) that cannot be observed for any other setting or estimator. The standard estimator $ST$ gives approximately correct coverage for $k \geq 20$ but is highly conservative for $k \leq 10$ studies, especially for five. $CR_3^*$ very consistently yields slightly more coverage than $CR_4^*$ in all settings except for $k=40$ where the difference between the two is negligible. For $k=5$ coverage based on $CR_4^*$ is approximately nominal and when based on $CR_3^*$ slightly conservative. For $k=10$ and $k=20$ $CR_4^*$ gives coverage around 91-92\% and $CR_3^*$ around 93-94\%. For $k=40$ both yield coverage around 92-94\%.

\begin{figure}[H]
\centering
\includegraphics[height=240pt,width=\textwidth]{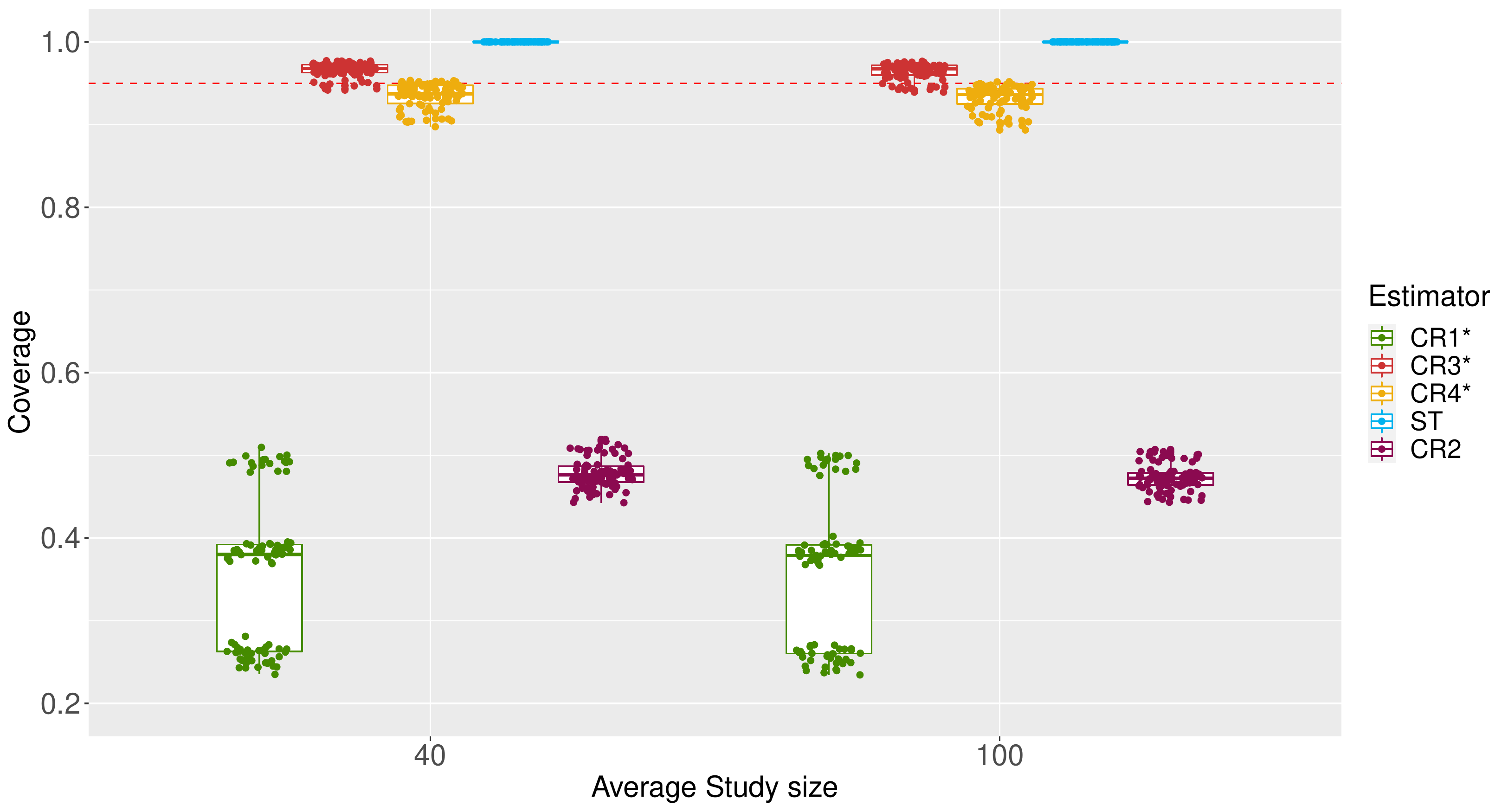}
\caption{Coverage of the confidence set (\ref{confidence_region}) based on an inversion of the adjusted $F$-test for $k=5$ studies.}
\label{fig:coverage5}
\end{figure}

\begin{figure}[H]
\centering
\includegraphics[height=240pt,width=\textwidth]{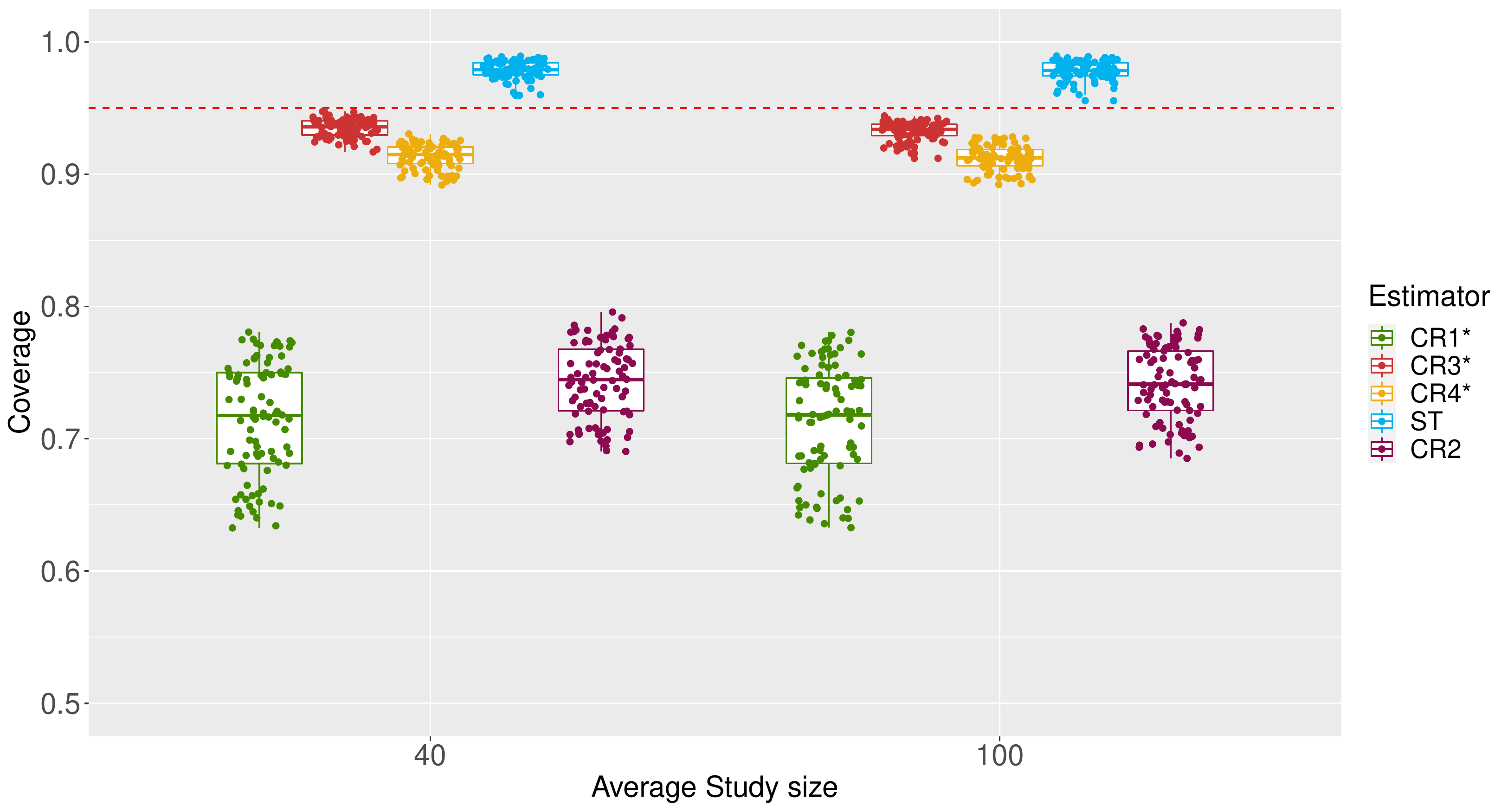}
\caption{Coverage of the confidence set (\ref{confidence_region}) based on an inversion of the adjusted $F$-test for $k=10$ studies.}
\label{fig:coverage10}
\end{figure}

\begin{figure}[H]
\centering
\includegraphics[height=240pt,width=\textwidth]{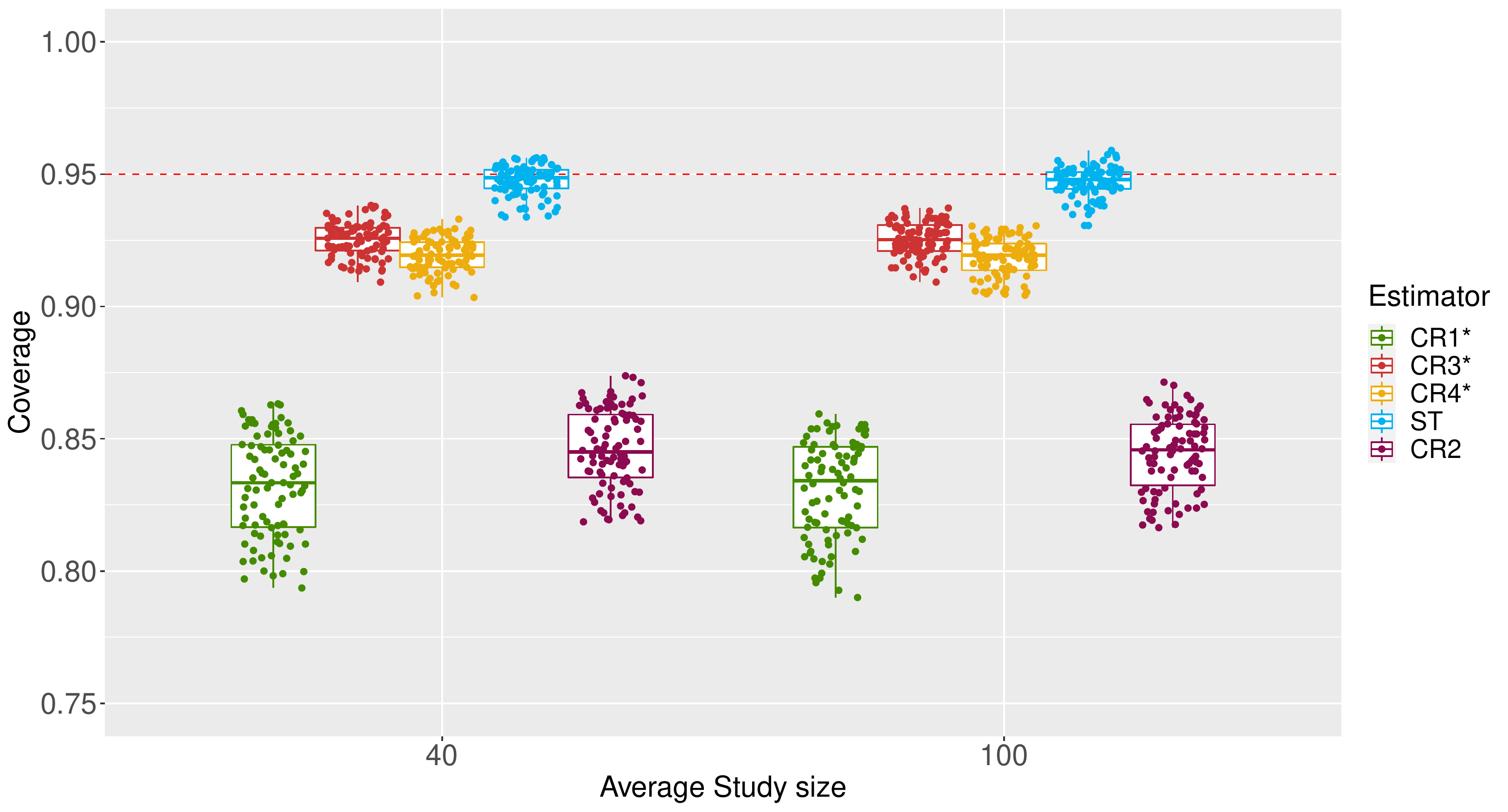}
\caption{Coverage of the confidence set (\ref{confidence_region}) based on an inversion of the adjusted $F$-test for $k=20$ studies.}
\label{fig:coverage20}
\end{figure}

\begin{figure}[H]
\centering
\includegraphics[height=230pt,width=\textwidth]{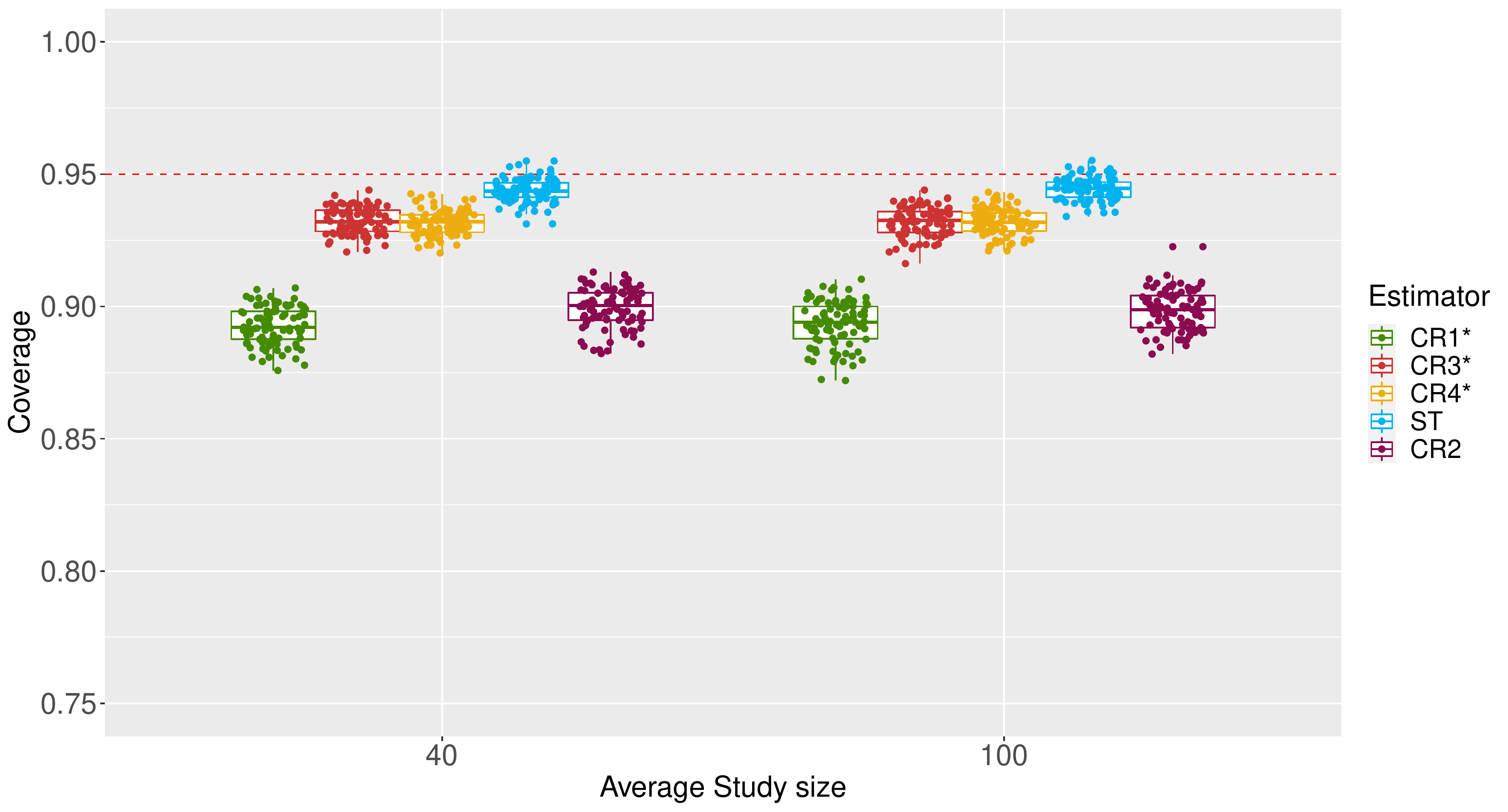}
\caption{Coverage of the confidence set (\ref{confidence_region}) based on an inversion of the adjusted $F$-test for $k=40$ studies.}
\label{fig:coverage40}
\end{figure}

In addition to these empirical coverage results, we also consider the power related to the respective tests and confidence regions. The power plots are provided in Figures \ref{fig:power1} and \ref{fig:power2} for $\bs{\beta} = (0.2,0.2,0.1,0.1)'$ and $\bs{\beta} = (0.4,0.4,0.2,0.3)'$ respectively. We show box plots to summarize the various simulation settings. For $\bs{\beta} = (0.4,0.4,0.2,0.3)'$ power is monotone increasing in the number of studies $k$ for all estimators. For $\bs{\beta} = (0.2,0.2,0.1,0.1)'$ power is monotone increasing in $k$ for $CR_3^*,CR_4^*$ and $ST$, whereas for $CR_1^*$ and $CR_2$ power decreases from a median of approximately 70\% and 60\% to 55\% and 52\% respectively, when going from five to ten studies and then increases in $k$ beyond this point.

The differences in power between the considered estimators are small for a large number of studies and become more pronounced as the number of studies decreases. For forty studies the power based on all estimators is nearly identical for both choices of $\bs{\beta}$. For twenty studies power based on $CR_1^*$ and $CR_2$ is slightly higher than for the other estimators. $CR_3^*$, $CR_4^*$ and $ST$ yield approximately the same power for both choices of $\bs{\beta}$ and twenty studies. For $k=10$ and $\bs{\beta} = (0.2,0.2,0.1,0.1)'$ the median power for $CR_1^*$ and $CR_2$ is around 55\% and 52\% respectively, whereas for $CR_3^*$, $CR_4^*$ and $ST$ it is around 25\%, 31\% and 20\% respectively. For $k=10$ and $\bs{\beta} = (0.4,0.4,0.2,0.3)'$ the median power for $CR_1^*$ and $CR_2$ is around 87\%, whereas for $CR_3^*$, $CR_4^*$ and $ST$ it is around 70\%, 74\% and 73\% respectively. For $k=5$ and $\bs{\beta} = (0.2,0.2,0.1,0.1)'$ the median power for $CR_1^*$ and $CR_2$ is around 70\% and 60\% respectively, whereas for $CR_3^*$, $CR_4^*$ and $ST$ it is only around 8\%, 12\% and 0\% respectively. For $k=5$ and $\bs{\beta} = (0.4,0.4,0.2,0.3)'$ the median power for $CR_1^*$ and $CR_2$ is around 83\% and 70\% respectively, whereas for $CR_3^*$, $CR_4^*$ and $ST$ it is around 13\%, 24\% and 1\% respectively.


\begin{figure}[H]
\centering
\includegraphics[height=210pt,width=\textwidth]{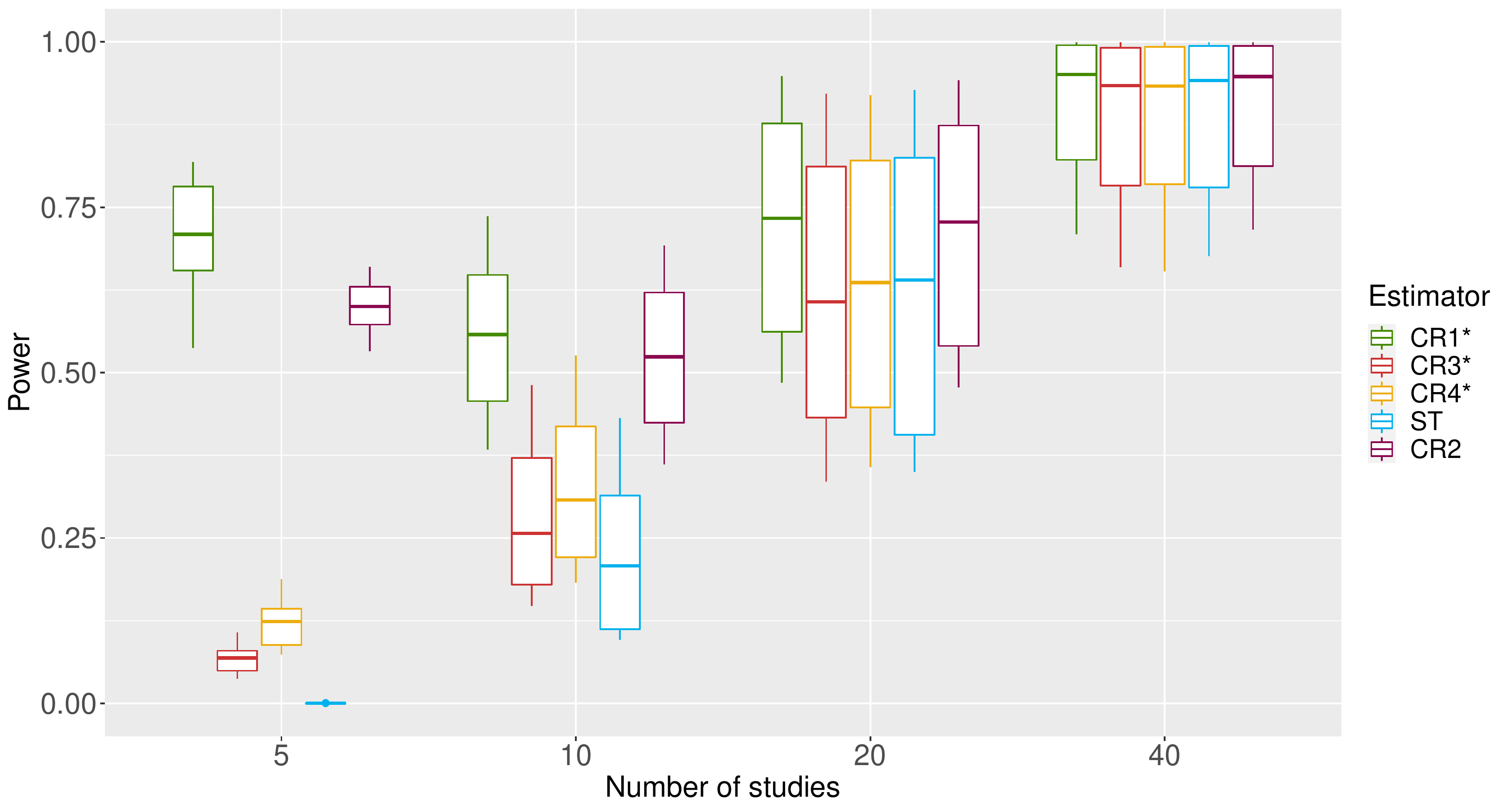}
\caption{Box plots of power based on adjusted $F$-test for all settings with $\bs{\beta} = (0.2,0.2,0.1,0.1)'$.}
\label{fig:power1}
\end{figure}

\begin{figure}[H]
\centering
\includegraphics[height=210pt,width=\textwidth]{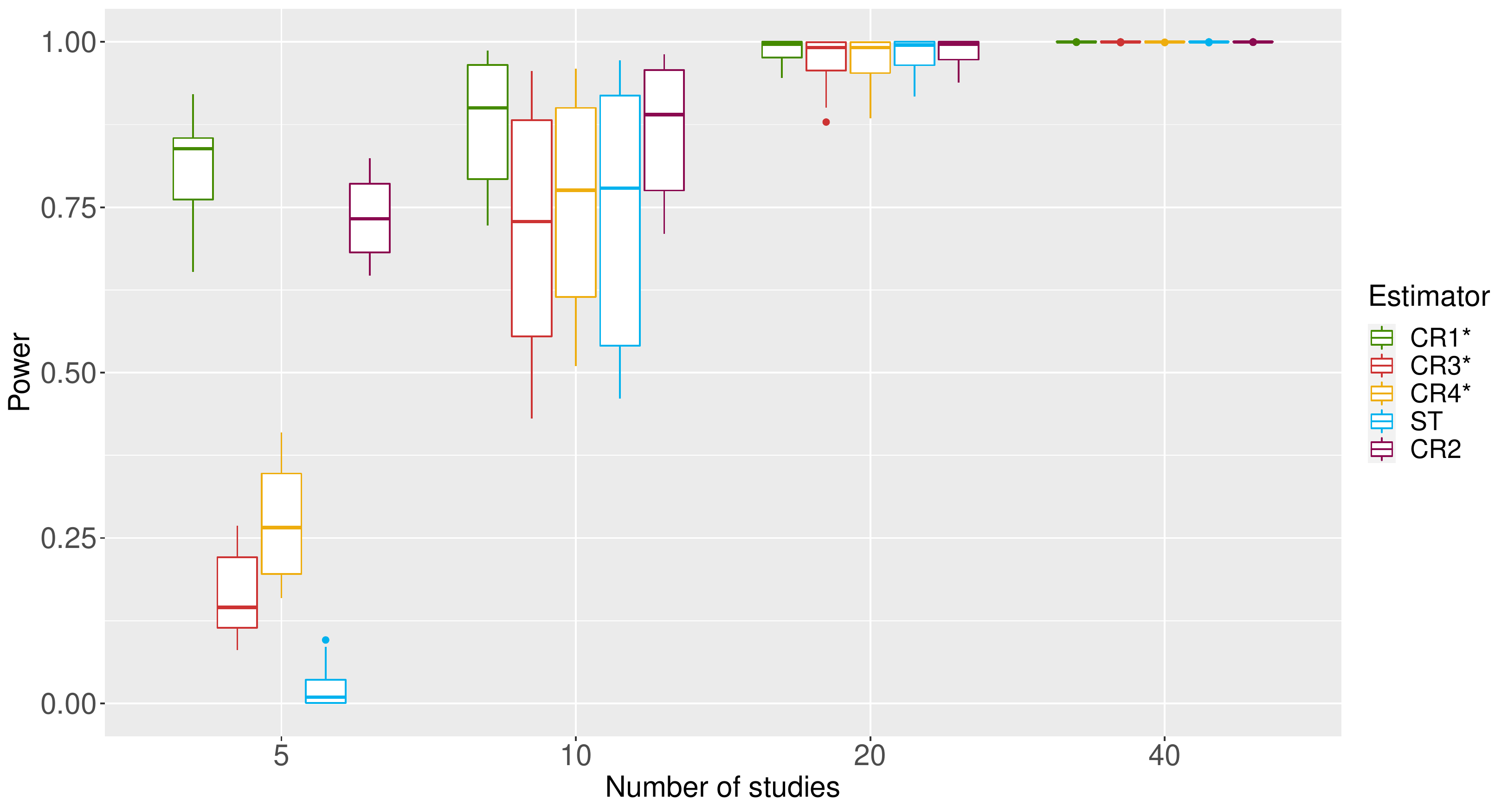}
\caption{Box plots of power based on adjusted $F$-test for all settings with $\bs{\beta} = (0.4,0.4,0.2,0.3)'$.}
\label{fig:power2}
\end{figure}

\newpage
\section{Discussion}

Multivariate Meta-Regression is an important tool for synthesizing and interpreting results from trials reporting multiple, correlated effects. However, information on these correlations is rarely available to analysts, making it difficult to construct the variance-covariance $\bs{V}$ matrix of the studies' sampling errors. Cluster-robust estimators allow for a correction of the standard errors, therefore enabling more reliable inference. In this paper we introduced two new proposals of CR estimators for use in multivariate meta-regression. We performed a simulation study, comparing these estimators with results based on two alternative CR estimators and the standard variance-covariance estimator with a focus on coverage and power of confidence sets and tests, as well as an illustrative real life data analysis. In our manuscript we only investigated the bivariate meta-regression setting, although all methods discussed are also applicable in higher dimensions. Further work is necessary to assess the viability of our suggestions in other settings, such as when the number of effects per study is greater than two.

Our main findings can be summarized as follows: The Zhang estimator, discussed in \cite{tipton2015smallPustejovsky}, can lead to a negative estimate of the denominator degrees of freedom in the $F$-distribution. This can occur when the number of studies is very small. The AHZ approach is therefore not recommendable for bivariate meta-regression if the number of studies is small ($k \leq 5$). Furthermore, when using the classical $F$-test in the bivariate setting, we recommend truncating the denominator degrees of freedom at two. The $CR_1^*$ and $CR_2$ estimators yield an empirical coverage that lies far below the nominal level $1 - \alpha$ and the coverage based on the other estimators, especially for smaller numbers of studies. On the flip side the tests based on these two $CR$-estimators unsurprisingly have superior power. The $ST$ estimator has approximately correct coverage for $k \geq 20$ studies but is highly conservative for $k \leq 10$ studies. $CR_3^*$ and $CR_4^*$ yield approximately correct coverage for five studies. $CR_3^*$ also gives nearly correct coverage for ten studies whereas $CR_4^*$ becomes slightly liberal in this case.

Based on our results we recommend using either the $CR_3^*$ or $CR_4^*$ estimator for bivariate meta-regression if $k \leq 10$ with a very slight preference for $CR_3^*$. For an analysis with $k \geq 20$ studies the $ST$ estimator seems to work best.

A limitation of our simulation study is that the sampling covariances between study-level effects were available for the construction of weight matrices. As mentioned in the introduction, this is often not feasible in practice, requiring analysts to calculate weights using a specified working model for the covariance structure. \cite{hedges2010robust} provide possible working models likely to be found in meta-analyses. They propose the use of approximately inverse variance weights, based on these working models.

An open question that requires further research is what the best testing procedure is when the number of studies $k$ is no greater than around five. Neither the adjusted Hotelling's $T^2$ approach in combination with Zhang's estimator for the degrees of freedom, which was recommended by \cite{tipton2015smallPustejovsky}, nor the naive or adjusted $F$-tests used in our simulations seem to be the ideal approach. This requires more intensive work that is outside the scope of this manuscript. For a discussion of alternative estimation approaches for the degrees of freedom in the adjusted Hotelling approach, we refer to \cite{tipton2015smallPustejovsky}. Another question for future research is whether other statistics or resampling approaches that have shown promising small sample approximations for heterogeneous MAN(C)OVA settings \citep{friedrich2017permuting,friedrich2018mats,zimmermann2020} can also help in multivariate meta-regression models.

\newpage
\bibliographystyle{apalike}
\bibliography{mvmeta_lit}

\newpage
\section*{Acknowledgements}

This work was supported by the German Research Foundation (DFG) (Grant no. PA-2409 7-1). The authors gratefully acknowledge the computing time provided on the Linux HPC cluster at TU Dortmund University (LiDO3), partially funded in the course of the Large-Scale Equipment Initiative by the German Research Foundation as project 271512359.

We would also like to thank James Pustejowsky for his helpful comments during the research phase for this manuscript.

\section*{Data Availability Statement}
The neuroblastoma dataset is contained in the \texttt{R} package \texttt{metafor}. All \texttt{R} scripts will be made publicly available, pending publication.

\newpage
\begin{center}
\large{Supplement to \enquote{Cluster-Robust Estimators for Bivariate Mixed-Effects Meta-Regression}}
\end{center}

\begin{theorem}
Suppose there is a $k_0 \in \mathbb{N}$ such that $k(\bs{X'\widehat{W}X})^{-1}$ exists and is uniformly bounded element-wise for all $k \geq k_0$. Furthermore, let $\widehat{\bs{T}}$ be a consistent estimator for $\boldsymbol{T}$ and $\Lambda$ the confidence region defined in the main paper. Then the $CR$ estimators $CR_0, CR_1, CR_2, CR_3, CR_3^*,$ $CR_4^*$ are asymptotically equivalent and we have $P(\Lambda \ni \bs{\beta}) \longrightarrow 1-\alpha$ as $k \rightarrow \infty$.
\end{theorem}

\begin{proof}
Let $k \in \mathbb{N}$ be the number of studies, $\bs{\beta} \in \R^{q}$ and $\textbf{X} \in \R^{k \times q}$. Furthermore, $\bs{\hat{\beta}} = \bs{(X'\widehat{W}X)^{-1}X'\widehat{W}y}$ and $\textbf{H} = \bs{X(X'\widehat{W}X)^{-1}X'\widehat{W}}$. Then
$$\textbf{H}\bs{X} = \bs{X(X'\widehat{W}X)^{-1}X'\widehat{W}X} = \textbf{X},$$
and since $\textbf{X}$ is a design matrix with the first column equal to $\bs{1}_{kp}$, all row sums in $\textbf{H}$ are equal to 1. Due to the regularity condition that there exists a $k_0 \in \mathbb{N}$ such that $\forall \ k \geq k_0: \ k(\bs{X'\widehat{W}X})^{-1}$ exists and is uniformly bounded element-wise, we have that for every $i,j \in \{1,\ldots,kp\}: h_{ij} \overset{a.s.}{\rightarrow} 0$ as $k \rightarrow \infty$.

So for $i \in \{1,\ldots,k\}$ it holds that $\textbf{H}_i \rightarrow \textbf{0}$ as $k \rightarrow \infty$. Here $\textbf{H}_i$ refers to the submatrix of $\textbf{H}$ with entries pertaining to study $i$. Thus $(\textbf{I}_{p_i}-\textbf{H}_i)^{\eta} \longrightarrow \textbf{I}_{p_i}$ and also $(\textbf{I}_{p_i}-\text{diag}(\textbf{H}_i))^{\eta} \longrightarrow \textbf{I}_{p_i}$ as $k \rightarrow \infty$ for any $\eta \in \R$. It follows that $\bs{\widehat{\Sigma}}_{a} - \bs{\widehat{\Sigma}}_{b} \rightarrow \textbf{0}_{q\times q}$ as $k \rightarrow \infty$ for any choice of $a,b \in \{CR_0,CR_1^*,CR_3,CR_3^*,CR_4^*\}$, i.e. they are asymptotically equivalent.

Consider the test statistic
$$Q = \bs{(\hat{\beta}-\beta_0)'\widehat{\Sigma}^{-1}_{CR}(\hat{\beta}-\beta_0)},$$

\noindent
where $\bs{\widehat{\Sigma}}_{CR}$ is one of the considered $CR$ variance-covariance estimators. Then for any choice of $CR$ estimator (as they are all consistent) we have $\bs{\widehat{\Sigma}}_{CR} \rightarrow \bs{\Sigma} = \text{Cov}(\bs{\widehat{\beta}})$ as $k \rightarrow \infty$. It follows with Slutzky's Lemma that $Q \overset{d}{\longrightarrow} \chi^2_q$ as $k \rightarrow \infty$ because with Lemma 2 in \cite{white1980heteroskedasticity}, it holds that 
$\bs{\widehat{\Sigma}^{-1/2}_{CR} \widehat{\beta}} \overset{d}{\longrightarrow} \mathcal{N}(\bs{\beta},\bs{I}_q)$. Furthermore it holds that $qF_{q,k-q,1-\alpha} \overset{d}{\longrightarrow} \chi^2_{q,1-\alpha}$ as $k \rightarrow \infty$ because $F_{q,k-q} \overset{d}{=} \frac{\chi_q^2/q}{\chi_{k-q}^2/(k-q)}$, where $\chi^2_q,\chi^2_{k-q}$ are independent chi-squared distributed random variables with $q,k-q$ degrees of freedom and $\xi := \chi^2_{k-q}/(k-q) \overset{a.s.}{\rightarrow} 1$ as $k \rightarrow \infty$ since $\mathbb{E}(\xi) \equiv 1$ and $\text{Var}(\xi) = \frac{2}{k-q} \rightarrow 0$ for $k \rightarrow \infty$. 

Therefore the confidence region $\Lambda$ from the main paper is an asymptotic $1-\alpha$ confidence region for $\bs{\beta}$.

\end{proof}

\end{document}